\let\normalvec\vec
\documentclass[oribibl,envcountsame,envcountsect]{llncs}
\let\vec\normalvec

\usepackage{amsmath,amssymb}
\usepackage[english]{babel}
\usepackage[T1]{fontenc}
\usepackage[ruled,noend,linesnumbered]{algorithm2e}
\usepackage{xspace}
\usepackage{graphicx}
\usepackage{subfigure}

\newenvironment{alg}{
  \begin{algorithm}[htbp]
    \SetKwInOut{Input}{input}
    \SetKwInOut{Output}{output}
  }{
  \end{algorithm}
}

\newcommand{\mc}[1]{\ensuremath{\mathcal{#1}}}
\newcommand{\bo}[1]{\mathcal{O}(#1)}
\newcommand{\conv}{\mathsf{conv}}
\newcommand{\real}{\mathbb{R}}
\newcommand{\nat}{\mathbb{N}}

\newcommand{\Soberon}{Sober\'{o}n}
\newcommand{\GarciaColin}{Garc\'{i}a-Col\'{i}n}
\newcommand{\Caratheodory}{Carath\'{e}odory}

\title{Algorithms for Tolerant Tverberg Partitions}
\author{Wolfgang Mulzer\thanks{Supported in part by DFG Grants MU 3501/1 and MU 3501/2.} \and Yannik Stein\thanks{
    Supported by the Deutsche Forschungsgemeinschaft within the research
    training group `Methods for Discrete Structures' (GRK 1408).
    }
}
\institute{ Institut f\"ur Informatik, Freie Universit\"at Berlin\\
            Takustr. 9, 14195 Berlin, Germany\\
            \email{\{mulzer, yannikstein\}@inf.fu-berlin.de}
          }
\begin{document}
\maketitle

\begin{abstract}
  Let $P$ be a $d$-dimensional $n$-point set. A partition $\mathcal{T}$ of $P$
  is called a \emph{Tverberg} partition if the convex hulls of all sets
  in $\mathcal{T}$ intersect in at least one point. We say that $\mathcal{T}$ is
  \emph{$t$-tolerant} if it remains a Tverberg partition after deleting
  any $t$ points from $P$. Sober\'{o}n and Strausz proved that there is always a
  $t$-tolerant Tverberg partition with $\lceil n / (d+1)(t+1) \rceil$ sets.
  However, no nontrivial algorithms for computing or approximating
  such partitions have been presented so far.

  For $d \leq 2$, we show that the Sober\'{o}n-Strausz bound can be improved,
  and we show how the corresponding partitions can be found in polynomial time.
  For $d \geq 3$, we give the first polynomial-time approximation algorithm
  by presenting a reduction to the regular Tverberg problem (with no
  tolerance).
  Finally, we show that it is coNP-complete to determine whether a given
  Tverberg partition is $t$-tolerant.

\end{abstract}

\section{Introduction}
Let $P\subset\real^d$ be a point set of size $n$. A point $c\in\real^d$ has
\emph{(Tukey) depth}  $m$ with respect to $P$ if every closed half-space containing
$c$ also contains at least $m$ points from $P$. A point of depth $\lceil n/(d+1)
\rceil$ is called a \emph{centerpoint} for $P$. The well-known centerpoint
theorem~\cite{Rado1946} states that every point set has a centerpoint. Centerpoints
are of great interest since they constitute a natural generalization of the median to
higher-dimensions and are invariant under scaling or translations
and robust against outliers.

Chan~\cite{Chan2004} described a randomized algorithm that
finds a $d$-dimensional centerpoint in expected time
$\bo{n^{d-1}}$. In fact, Chan solves the seemingly harder problem of finding a
point with maximum depth, and he conjectures that his result is optimal.
Since this is infeasible in higher
dimensions, approximation algorithms are of interest. Already in 1993, Clarkson et
al.~\cite{Clarkson1996} developed a Monte-Carlo algorithm that finds a point
with depth $\Omega(n/(d+1)^2)$ in time $\bo{d^2(d\log n +
  \log(1/\delta))^{\log(d+2)}}$, where $\delta$ is the error-probability.
Teng~\cite{Teng1992} proved that it is coNP-complte to test whether
a given point is a centerpoint. Thus, we do not know how to verify
efficiently the output of the algorithm by Clarkson et al. For a subset
of centerpoints, \emph{Tverberg partitions}~\cite{Tverberg1966} provide
polynomial-time checkable proofs for the depth: a \emph{Tverberg $m$-partition} for a
point set $P\subset \real^d$ is a partition
$P = T_1 \dot\cup T_2 \dot\cup \cdots \dot\cup T_m$ of $P$ into $m$ sets such
that $\bigcap_{i=1}^{m} \conv(T_i)\neq \emptyset$.  Each half-space that
intersects $\bigcap_{i=1}^{m} \conv(T_i)$ must contain at least one
point from each $T_i$, so each point in $\bigcap_{i=1}^{m} \conv(T_i)$ has
depth at least $m$.
Tverberg's theorem states that depth $m = \lceil n/(d+1) \rceil$ can always be
achieved.
Thus, there is always a centerpoint with a corresponding Tverberg partition.
Miller and Sheehy~\cite{Miller2010} developed
a deterministic algorithm that computes a
point of depth $\lceil n/2(d+1)^2\rceil$ in time $n^{\bo{\log d}}$ together
with a corresponding Tverberg partition. This was recently improved by Mulzer
and Werner~\cite{Mulzer2013}. Through recursion on the dimension, they can find
a point of depth $\lceil n/4(d+1)^3 \rceil$ and a corresponding Tverberg partition
in time $d^{\bo{\log d}} n$.

\begin{figure}[htbp]
  \newcommand{\imgwidth}{0.42\textwidth}
  \begin{center}
    \subfigure[$0$-tolerant Tverberg $5$-partition]{
        \label{intro:fig:tolex:untol}\includegraphics[width=\imgwidth]{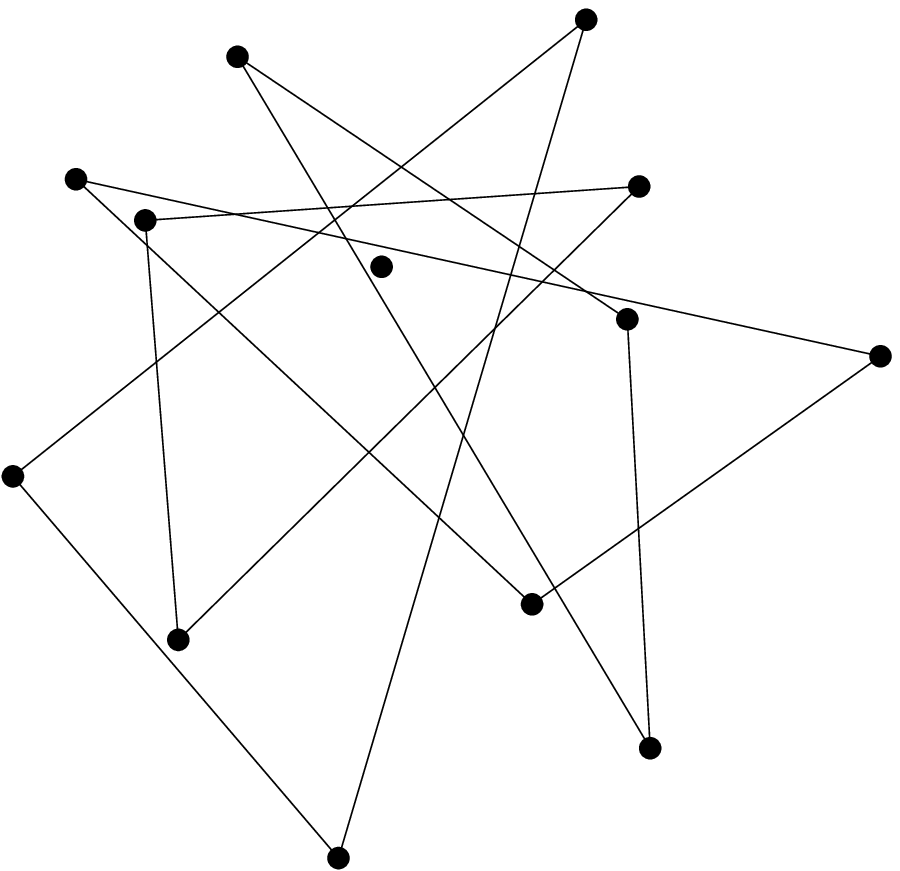}
    }
    \hspace{1cm}
    \subfigure[$1$-tolerant Tverberg $3$-partition]{\label{intro:fig:tolex:tol}\includegraphics[width=\imgwidth]{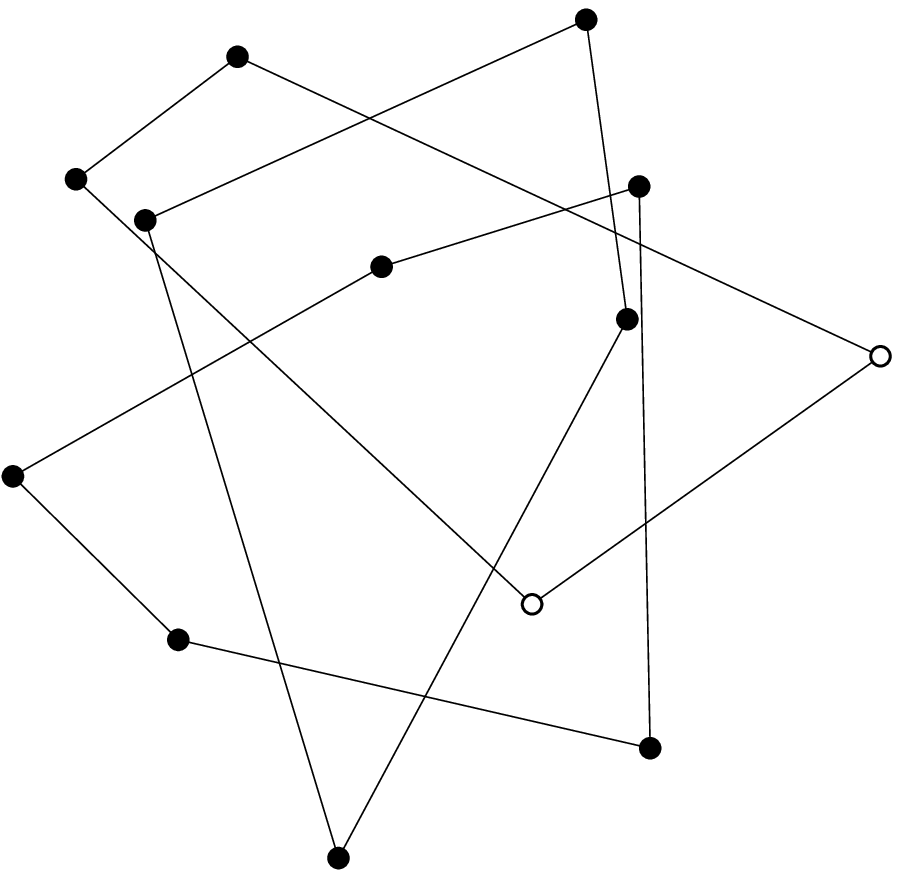}}
  \end{center}
  \caption{
    \subref{intro:fig:tolex:untol} A regular Tverberg partition
    for 13 points. One set of the partition consists of
    a single point. The removal of any point would separate the
    convex hulls.
    \subref{intro:fig:tolex:tol} A $1$-tolerant Tverberg partition
    for the same point set. The Tverberg partition is not
    $2$-tolerant since the removal of both white points would
    separate the convex hulls.
  }
  \label{intro:fig:tolex}
\end{figure}

Let $\mc{T}$ be a Tverberg $m$-partition for $P$. If any nonempty subset
$R\subset P$ is removed from $P$, we no longer know if $\bigcap_{i=1}^m
\conv(T_i \setminus R) \neq \emptyset$. In the worst-case, the maximum number of
sets in $\mc{T}$ whose convex hulls still have a nonempty intersection is
$m-|R|$. Thus, after removing only $m$ points, the convex hulls of sets in $\mc{T}$
could be pairwise disjoint and do not longer serve as a depth-certificate for
points in the intersection of the convex hulls. In order to give stronger
guarantees if
points are removed, we study Tverberg
partitions that remain Tverberg partitions even after removing $t$ arbitrary
points from $P$. We call a Tverberg partition \emph{$t$-tolerant}
if $\bigcap_{i=1}^m \conv(T_i \setminus R)$ is nonempty for any subset $R\subset
P$ of size at most $t$. To distinguish tolerant Tverberg partitions from
Tverberg partitions with no tolerance, we  call the latter \emph{regular}
Tverberg partitions. Furthermore, we say that a set $R$ \emph{separates} the convex hulls of the
sets in $\mc{T}$ if $\bigcap_{i=1}^m \conv(T_i \setminus R) = \emptyset$. See
Figure~\ref{intro:fig:tolex} for two examples.
In 1972, Larman~\cite{Larman1972} proved that
every set of size $2d+3$ admits a $1$-tolerant Tverberg
$2$-partition.  This was motivated by a problem proposed by
McMullen: find the largest number $n$ such that any $n$-point set
can be made convex by applying a permissible projective
transformation. Here, permissible means
that no point in the set is mapped to a point at infinity.
Larman also showed that there are sets of size $d+\Theta(\sqrt{d})$
that do not have a $1$-tolerant Tverberg
$2$-partition if $d\geq 2$. This lower bound was later improved by
Ram\'irez Alfons\'in~\cite{Alfonsin2001} to
$5/3 d + 4/3$ for $d\geq 4$. \GarciaColin{}~\cite{GarciaColin2007} generalized
Larman's upper bound, showing that sets of size $(t+1)(d+1)+1$ always have a
$t$-tolerant Tverberg $2$-partition, and asked for a general bound to guarantee the
existence of $t$-tolerant Tverberg $m$-partitions. Later, Montejano and
Oliveros~\cite{Montejano2011} conjectured that every set of size $(t+1)(m-1)(d+1)+1$ admits a
$t$-tolerant Tverberg $m$-partition. This
was proved by \Soberon{} and Strausz~\cite{Soberon2012} who adapted Sarkaria's
proof of Tverberg's theorem~\cite{Sarkaria1992} to the tolerant setting:

\begin{theorem}[\Soberon{}-Strausz-Theorem~\cite{Soberon2012}]\label{thm:sobstrau}
  Let $P\subset \real$ be a set of size $(t+1)(m-1)(d+1)+1$. Then, there exists
  a $t$-tolerant Tverberg $m$-partition for $P$.

  Equivalently, for all $n$-point sets there
  exists a $t$-tolerant Tverberg $\lceil n/(d+1)(t+1)\rceil$-partition.
\end{theorem}

\Soberon{} and Strausz~\cite{Soberon2012} also
conjectured this bound to be tight. A lower bound was recently
proven by \Soberon{}~\cite{Soberon2014}: at least
$m(\lfloor d/2\rfloor + t + 1)$ points are necessary to guarantee the
existence of a $t$-tolerant Tverberg $m$-partition.

So far,
no exact or approximation algorithms for tolerant Tverberg partitions appear in
the literature.

\paragraph{Our contribution.}
\newenvironment{prfTODO}{\begin{proof}[TODO]}{\qed\end{proof}}

In Section~\ref{sec:ld}, we consider the problem of computing
tolerant Tverberg partitions in low dimensions. We present an
algorithm for the one-dimensional case and use a dimension-reduction
argument to extend the algorithm to multidimensional input:
\begin{theorem}\label{thm:dr_tverberg}
  Given a set $P\subset \real^{d}$ of size $2^{d-1}(m(t+2) - 1)$,
  a $t$-tolerant Tverberg $m$-partition
  for $P$ can be computed in time $\bo{2^{d-1}dmt + mt \log t}$.
\end{theorem}
For $d=1$, the bound on the number of points is tight and improves
the \Soberon{}-Strausz bound from Theorem~\ref{thm:sobstrau} by
$t(m-2)$. For $d=2$, the new bound improves the bound
by \Soberon{} and Strausz for large enough $m$ and $t$.

For higher dimensions, we describe in Section~\ref{sec:redregtver} an
approximation-preserving reduction to the regular Tverberg problem based on a
lemma by \GarciaColin{}. Thus, we
can apply existing and possible future algorithms for the regular Tverberg
problem in the tolerant setting:

\begin{proposition} \label{apr:prop:cb}
  Let $P\subset \real^d$ and let $\mc{A}$ be an
  algorithm that computes a regular Tverberg $m$-partition
  for any point set of size $n_\mc{A}(m)$ in time $T_{\mc{A}}(m)$.
  Then, a $\left(\lfloor |P|/n_{\mc{A}}(m) \rfloor -1\right)$-tolerant Tverberg
  $m$-partition for $P$ can be computed in time
  $\mc{O}\left(T_{\mc{A}}(m) \cdot |P| / n_{\mc{A}}(m) \right)$.
\end{proposition}

Finally, we show in Section~\ref{sec:complexity} that it is
coNP-complete to determine whether a given Tverberg partition has
tolerance $t$ if the dimension is part of the input:
\begin{theorem}\label{thm:complexity}
  \textsc{TestingTolerantTverberg} is coNP-complete.
\end{theorem}
This holds even if we restrict the input to Tverberg partitions of size 2.

\section{Low Dimensions} \label{sec:ld}
We start with an algorithm for the one-dimensional case that yields a
tight bound. This can be bootstrapped to higher dimensions with a lifting approach
similar to the algorithm by Mulzer and Werner~\cite{Mulzer2013}. In two dimensions, we also get
an improved bound if the size of the desired partition and the
tolerance is large enough.

\subsection{One Dimension}
Let $P\subset \real$ with $|P| = n$, and let
$\mc{T}=\{T_1,T_2,\ldots,T_m\}$ be a $t$-tolerant Tverberg $m$-partition of $P$. By
definition, there is no
subset $R\subset P, |R| = t$ whose removal separates the convex hulls of the
sets in $\mc{T}$. Bounding the size of the sets in \mc{T} gives us more insight
into the structure.

\begin{lemma}\label{lem:1d_sizes}
  Let $P\subset\real$ with $|P| = n$ and let $\mc{T}=\{T_1,T_2,\ldots,T_m\}$ be
  a $t$-tolerant Tverberg $m$-partition of $P$. Then,
  \begin{enumerate}
    \renewcommand{\labelenumi}{(\roman{enumi})}
    \item for $i = 1, \ldots,m$, we have $|T_i| \geq t+1$; and
    \item for $i,j = 1, \ldots,m$, $i\neq j$, we have $|T_i \cup T_j| \geq 2t+3$.
  \end{enumerate}
\end{lemma}
\begin{proof}
  \begin{enumerate}
    \renewcommand{\labelenumi}{(\roman{enumi})}
    \item Suppose $|T_i| \leq t$. After removing $T_i$ from $P$, the
          intersection of the convex hulls of the sets in $\mc{T}$ becomes empty,
          and $\mc{T}$ would not be $t$-tolerant.
    \item Suppose there are  $T_i, T_j\in\mc{T}$ with $|T_i \cup
    T_j| \leq 2t+2$. By $(i)$, we have $|T_i|=|T_j|=t+1$.
    Let $p_{\min} = \min(T_i\cup T_j)$ and assume w.l.o.g. that $p_{\min} \in T_i$
    (see Figure~\ref{fig:1d_sizes}).
    Then $|T_i \setminus \{p_{\min}\}| = t$, and removing the set $T_i
    \setminus \{p_{\min}\}$ separates the convex hulls of $T_i$ and $T_j$.
    This again contradicts  $\mc{T}$ being $t$-tolerant.
  \end{enumerate}
\end{proof}

\begin{figure}[htbp]
  \begin{center}
    \includegraphics{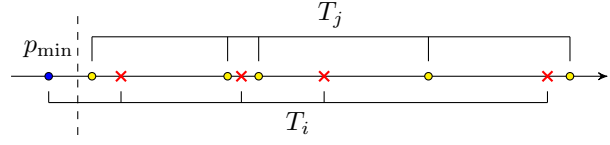}
  \end{center}
  \caption{The convex hulls of two sets of size $t+1$ can be separated by
    removing $t$ points.}
  \label{fig:1d_sizes}
\end{figure}

Lemma~\ref{lem:1d_sizes} immediately implies a lower bound on the size of any point set
that admits a $t$-tolerant Tverberg $m$-partition.

\begin{corollary}
  \label{ld:cor:lb}
  Let $P\subset \real$ with $|P| < m(t+2)-1$. Then $P$ has no $t$-tolerant
  Tverberg $m$-partition.
\end{corollary}

Now what happens for $|P| = m(t+2)-1$? Note that for $t > 0$ and $m > 2$, we
have $m(t+2) -1 < 2(t+1)(m-1) + 1$, the bound by \Soberon{} and Strausz.
Thus, proving that a $t$-tolerant Tverberg $m$-partition exists for any
one-dimensional point set of size $m(t+2)-1$ would disprove the conjecture by
\Soberon{} and Strausz.

Let $P\subset \real$ be of size $m(t+2)-1$. By Lemma~\ref{lem:1d_sizes}, in any
$t$-tolerant Tverberg partition of $P$, one set has to be of size $t+1$ and all
other sets have to be of size $t+2$. Let $\mc{T}=\{T_1,\ldots,T_m\}$ be
a Tverberg $m$-partition of $P$ such that $T_1$ contains every $m$th
point of $P$ and each other set $T_i$ ($i\geq 2$) has one point in each interval
defined by the points of $T_1$; see Fig.~\ref{ld:fig:ttp} for $m =3$ and $t=2$.
Note that $|T_1| = t+1$ and $|T_i| = t+2$ for $i\geq 2$.
We will show that $\mc{T}$ is $t$-tolerant.
Intuitively, $\mc{T}$
maximizes the interleaving of the sets, making the convex hulls more
robust to removals.

\begin{figure}[htbp]
  \begin{center}
    \includegraphics[width=0.7\textwidth]{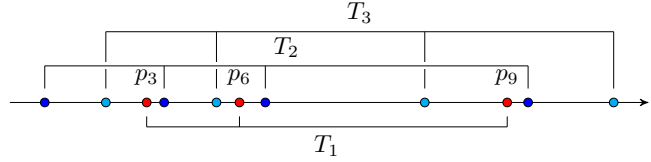}
  \end{center}
  \caption{A $2$-tolerant Tverberg $3$-partition for $11$ ($=3(2+2)-1$)
    points.}
  \label{ld:fig:ttp}
\end{figure}

\begin{lemma}
  \label{ld:lem:1d}
  Let $P\subset \real$ with $|P| = m(t+2)-1$, and let
  $\mc{T}=\{T_1,\ldots,T_m\}$ be an $m$-partition of $P$.
  Suppose that $|T_1| = t+1$, and write $T_1 = (p_1, p_2, \dots, p_{t+1})$, sorted
  from left to right. Suppose that each interval $\mc{I}\in \{(-\infty, p_{1})$,
  $(p_{1},p_{2})$, $\ldots$,$(p_{t+1},\infty)\}$ contains one point
  from each $T_i$, for $i = 2, \dots, m$. Then $\mc{T}$
  is a $t$-tolerant Tverberg $m$-partition for $P$.
\end{lemma}
\begin{proof}
  Suppose there exist $T_i, T_j \in \mc{T}$, $i \neq j$, and a subset $R\subset
  P$ of size $t$ such that removing $R$  from $P$ separates the convex hulls of
  $T_i$ and $T_j$.  Let $h$ be a point that separates $\conv(T_i \setminus R)$
  and $\conv(T_j \setminus R)$.  Define $a = \max\{p \in T_1 \mid p\leq h\}$ and
  $b= \min\{p \in T_1 \mid p > h\}$, where $a = -\infty$ if all points in $T_1$
  are greater than $h$ and $b=+\infty$ if all points in $T_1$ are less than $h$.
  Let $T_i^{\leq a} = \{ p \in T_i \mid p \leq a\}$ and $T_i^{\geq b} = \{p \in
  T_i \mid p \geq  b\}$, and define $T_j^{\leq a}$, $T_j^{\geq b}$ similarly.
  Since removing $R$ separates the convex hulls of $T_i$ and $T_j$ at $h$, $R$
  must contain either $T_i^{\leq a} \cup T_j^{\geq b}$ or $T_i^{\geq b} \cup
  T_j^{\leq a}$. Figure~\ref{ld:fig:lem1dproof} shows the situation. By
  construction, we know that $|T^{\leq a}_i| = |T^{\leq a}_j|
  = |T^{\leq a}_1|$ and $|T^{\geq b}_i| = |T^{\geq b}_j| = |T^{\geq b}_1|$. We
  thus have $t \geq |R| \geq |T_i^{\leq a} \cup T_j^{\geq b}| = |T_j^{\leq a}
  \cup T_i^{\geq b}| = |T_1^{\leq a} \cup T_1^{\geq b}|$. However, since
  $|T_1^{\leq a} \cup T^{\geq b}_1| = |T_1| = t+1$, this is a contradiction.

  Thus, even after removing $t$ points, the convex hulls of the sets in $\mc{T}$
  intersect pairwise. Helly's theorem~\cite{Matouvsek2002} now
  guarantees that the convex hulls of all sets in $\mc{T}$ have a common
  intersection point. Hence, $\mc{T}$ is $t$-tolerant.
\end{proof}

\begin{figure}[htbp]
  \begin{center}
    \includegraphics[width=0.6\textwidth]{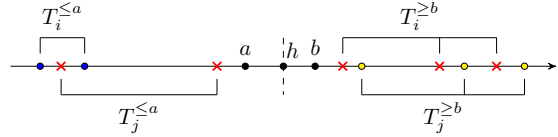}
  \end{center}
  \caption{The convex hulls of two elements in \mc{T} are separated after the
    removal of $R$. Crosses mark removed points (i.e., points in $R$).
    Points not used in the proof of Lemma~\ref{ld:lem:1d} are left out for
    clarification.}
  \label{ld:fig:lem1dproof}
\end{figure}

Lemma~\ref{ld:lem:1d} immediately gives a way to compute a $t$-tolerant Tverberg
$m$-partition in $\bo{mt \log mt}$ time for $|P|=m(t+2)-1$ by sorting
$P$. However, it is not necessary to know the order of all of $P$.
Algorithm~\ref{ld:alg:1d} exploits this fact to improve the running time. It
repeatedly partitions the point set until it has selected all points whose
ranks are multiples of $m$. These points form the set $T_1$.
Initially, the set $Q$
contains only the input $P$ (line \ref{ld:alg:1d:qinit}). In lines
\ref{ld:alg:1d:splitstart}--\ref{ld:alg:1d:splitend}, we select from each set in
$Q$ an element whose rank is a multiple of $m$ (line \ref{ld:alg:1d:select}) and
we split the set at this element. Here, $\mathtt{select}(P,k)$ is a procedure that
returns the element with rank $k$ of $P$. After termination of both loops
in lines \ref{ld:alg:1d:outerstart}--\ref{ld:alg:1d:outerend}, all remaining sets
in $Q$ correspond to points in $P$ between two consecutive points in $T_1$. In
lines \ref{ld:alg:1d:fstart}--\ref{ld:alg:1d:fend}, the points in the sets in
$Q$ are distributed equally among the elements $T_i$ ($i\geq 2$) of the returned
partition.

\begin{alg}
  \Input{$P\subset \real$, size of partition $m$}
  \SetKwFunction{Select}{select}
  $r \leftarrow m$\;\label{ld:alg:1d:rstart}
  \While{$r \leq |P|/2$}{
    $r \leftarrow 2 \cdot r$\;\label{ld:alg:1d:se}\label{ld:alg:1d:rend}
  }
  $Q \leftarrow \{ P \}$;\label{ld:alg:1d:qinit}
  $T_1,T_2,\ldots,T_m \leftarrow \emptyset, \emptyset,\ldots, \emptyset$\;
  \While{$r \geq m$}{
    \label{ld:alg:1d:outerstart}
    \ForEach{$P' \in Q \text{ with } |P'| \geq r$}{\label{ld:alg:1d:splitstart}
        remove $P'$ from $Q$\;
        $p_r \leftarrow \Select(P', r)$\;\label{ld:alg:1d:select}
        $Q \leftarrow Q \cup \{ \{ p' \in P' \mid p' < p_r \}, \{ p' \in P' \mid p' >
          p_r \}\}$\;\label{ld:alg:1d:qadd}
        $T_1 \leftarrow T_1 \cup \{ p_r \}$\;\label{ld:alg:1d:splitend}
    }
	$r \leftarrow r/2$\;\label{ld:alg:1d:r}
    \label{ld:alg:1d:outerend}
  }
    \ForEach{$P' \in Q$}{\label{ld:alg:1d:fstart}
  \ForEach{$j \in \{2,3,\ldots,m\}$}{
      remove any point from $P'$ and add it to $T_j$\; \label{ld:alg:1d:fend}
      }
    }
  \Return{$\{T_1,T_2,\ldots,T_m\}$}\;
  \caption{1d-tolerant-Tverberg}
  \label{ld:alg:1d}
\end{alg}

\begin{theorem}
  Let $P\subset\real$ be a set of size $m(t+2)-1$. On input $(P,m)$,
  Algorithm~\ref{ld:alg:1d} returns a $t$-tolerant Tverberg partition for $P$ in
  time $\bo{mt \log t}$.
  \label{ld:thm:bound}
\end{theorem}
\begin{proof}
  After each complete run of the inner
  for-loop (lines \ref{ld:alg:1d:splitstart}--\ref{ld:alg:1d:splitend}),
  each element $P'\in Q$ has size strictly less than
  $r$:
  initially, $Q$ contains only $P$ and $r$ is strictly greater than $|P|/2$.
  Hence, both new sets added to $Q$ in line~\ref{ld:alg:1d:qadd}
  are of size strictly less than $r$.
  Since $r$ is halved after each run (line~\ref{ld:alg:1d:r}), the invariant is maintained.

  We will now check that Lemma~\ref{ld:lem:1d} applies. We only split the sets
  in $Q$ at elements whose rank is a multiple of $m$, so the ranks do not change
  modulo $m$. By the invariant, after the termination of the outer while-loop
  in lines \ref{ld:alg:1d:outerstart}--\ref{ld:alg:1d:outerend}, each set in $Q$
  has size strictly less than $m$. Thus, $T_1$ contains the
  points in $P$ whose rank is a multiple of $m$ and
  each set $P'\subset P$ in $Q$ contains all points between two consecutive
  points in $T_1$. Since these are distributed equally among
  $T_2,\ldots,T_m$, Lemma~\ref{ld:lem:1d} now shows the correctness of the
  algorithm.

  Let us consider the running time. Computing the initial $r$ in
  lines~\ref{ld:alg:1d:rstart}--\ref{ld:alg:1d:rend} requires
  $\bo{\log(|P|/m)}=\bo{t}$ time. The split-element in
  line~\ref{ld:alg:1d:select} can be found in time
  $\bo{|P'|}$~\cite{Cormen2009}. Thus, since the sets are disjoint,
  one iteration of the outer while-loop requires $\bo{|P|}$ time, for a total of
  $\bo{\log(|P|/m) |P|}=\bo{\log(t) mt}$. By the same argument, both for-loops
  in lines \ref{ld:alg:1d:fstart}--\ref{ld:alg:1d:fend} require linear time in
  the size of $P$. This results in a total time complexity of $\bo{mt \log t}$,
  as claimed.
\end{proof}

\subsection{Higher Dimensions}
We use a lifting argument~\cite{Mulzer2013} to extend Algorithm~\ref{ld:alg:1d} to
higher-dimensional input.
Given a point set $P\subseteq \real^{d}$ of size $n$, let $h$ be a hyperplane that
splits $P$ evenly (if $n$ is odd, $h$ contains exactly one point of $P$). We
then partition $P$ into $\lfloor n/2 \rfloor$ pairs $(p_i^-,p_i^+)$, where
$p_i^-\in h^-$ and $p_i^+\in h^+$. We obtain a $(d-1)$-dimensional point set
with $\lfloor n/2 \rfloor$ elements by mapping each pair to the intersection of
the connecting line segment with $h$.

Let $q_i = p_i^+p_i^- \cap h$ be the mapped point for $(p_i^-,p_i^+)$ and
$\mathcal{T'}=\{T'_1, \ldots,T'_m\}$ a
$t$-tolerant Tverberg $m$-partition of $Q=\{q_1,\ldots,q_{\lfloor n/2 \rfloor}\}$.
We obtain a Tverberg $m$-partition $\mathcal{T}$ with tolerance $t$ for
$P$ by replacing each $q_i$ in $\mathcal{T}'$ by its corresponding pair
$(p_i^-,p_i^+)$. Thus, we can repeatedly project the set $P$
until Algorithm~\ref{ld:alg:1d} is applicable. Then, we lift the one-dimensional
solution back to higher dimensions.

Algorithm~\ref{ld:alg:dred} follows this approach.
For $d=1$, Algorithm~\ref{ld:alg:1d} is applied
(lines~\ref{ld:alg:dred:basestart}--\ref{ld:alg:dred:baseend}).
Otherwise, we take an appropriate hyperplane orthogonal to the $x_d$-axis and
compute the lower-dimensional point set
(lines~\ref{ld:alg:dred:h}--\ref{ld:alg:dred:pairend}). This is always possible
since we can assume w.l.o.g.
that all points have distinct $x_d$ coordinates by a simple rotation argument.
Finally, the result for $d-1$ dimensions is lifted back to $d$ dimensions
(lines~\ref{ld:alg:dred:liftstart}--\ref{ld:alg:dred:liftend}).
\begin{alg}
  \Input{point set $P\subset\real^{d}$, tolerance parameter $t$,
  size of partition $m$}
  \Output{$t$-tolerant Tverberg partition for P of size $m$}
  \SetKwFunction{solveOned}{1d-tolerant-Tverberg}
  \SetKwFunction{rec}{DimReduct-Tolerant-Tverberg}
  \If{$d=1$}{\label{ld:alg:dred:basestart}
    \Return{\solveOned$(P, m)$}
    \label{ld:alg:dred:baseend}
  }
  $h \leftarrow$ hyperplane that halves $P$ according to the $x_d$-coordinate\;
  \label{ld:alg:dred:h}
  \ForEach{$i\in \{1,2,\ldots,|P\cap h^-|\}$}{
    \label{ld:alg:dred:pairstart}
    $p_i^- \leftarrow$ remove any point from $P$ that belongs to $P\cap h^-$\;
    $p_i^+ \leftarrow$ remove any point from $P$ that belongs to $P\cap h^+$\;
    $q_i \leftarrow$ first $d-1$ coordinates of $p_i^-p_i^+\cap h$\;
    \label{ld:alg:dred:pairend}
  }
  $Q \leftarrow \{q_1,q_2,\ldots,q_{|P\cap h^-|}\}$\;
  $\{T'_1,T'_2,\ldots,T'_m\} \leftarrow \rec(Q, t, m)$\;
  \label{ld:alg:dred:reccall}
  \ForEach{$j \in \{1,2,\ldots,m\}$}{
    \label{ld:alg:dred:liftstart}
    $T_j \leftarrow \{ p_i^-,p_i^+ \mid q_i \in T'_j\}$\;
    \label{ld:alg:dred:liftend}
  }
  \label{ld:alg:dred:return}
  \Return{$\{T_1,T_2,\ldots,T_m\}$}\;
  \caption{DimReduct-Tolerant-Tverberg}
  \label{ld:alg:dred}
\end{alg}
Using Theorem~\ref{ld:thm:bound}, it is easy to show that
Algorithm~\ref{ld:alg:dred} achieves the bounds claimed in Theorem~\ref{thm:dr_tverberg}:
\begin{theorem}[Theorem~\ref{thm:dr_tverberg} restated]
  Given a set $P\subset \real^{d}$ of size $2^{d-1}(m(t+2) - 1)$,
  a $t$-tolerant Tverberg $m$-partition
  for $P$ can be computed in time $\bo{2^{d-1}dmt + mt \log t}$.
\end{theorem}
\begin{proof}
  Since the size of $P$ halves in each recursion step, $2^{d-1}(m(t+2) - 1)$ points
  suffice to ensure that Algorithm~\ref{ld:alg:1d} can be applied to $m(t+2)-1$
  points in the base case. Each projection and lifting step can be performed in linear
  time, using a median computation. Since the size of the point set decreases
  geometrically, the total time for projection and lifting is thus $\bo{2^{d-1}dmt}$.
  Since Algorithm~\ref{ld:alg:1d} has running time $O(mt\log t)$, the result follows.
\end{proof}

For $d\geq 3$, the bound from Proposition~\ref{thm:dr_tverberg}
is worse than the \Soberon{}-Strausz bound. However, in two dimensions, we have
\[
  2^{2-1}(m(t+2)-1) <  (2+1)(m-1)(t+1) + 1 \Leftrightarrow m/(m-3) <   t
\]

This holds for instance if $ m\geq 4 \wedge t\geq 5 $ or $m\geq 7 \wedge t\geq 2$.
Thus, Algorithm~\ref{ld:alg:dred} gives a strict improvement over
the \Soberon{}-Strausz bound for large enough $m$ and $t$.

\section{Reduction to the Regular Tverberg Problem}\label{sec:redregtver}

We now show how to use any algorithm that computes approximate
regular Tverberg partitions in order to find tolerant Tverberg partitions.
For this, we must increase the tolerance of a Tverberg partition.
In the following, we show that one can merge elements of several
Tverberg partitions for disjoint subsets of $P$
to obtain a Tverberg partition with higher tolerance for the whole set $P$.
The following lemma is also implicit in the Ph.D. thesis of
\GarciaColin{}~\cite{GarciaColin2007}.

\begin{lemma}
  \label{apr:lem:cb}
  Let $\mc{T}_1,\ldots,\mc{T}_k$ be Tverberg $m$-partitions
  for disjoint point sets $P_1$, $\ldots$,$ P_k\subset\real^d$.
  Let $T_{i,j}$ be the $j$th element of $\mc{T}_i$ and $t_i \geq 0$
  the tolerance of $\mc{T}_i$.
  Then,
    $\mc{T} = \{ T_j = \textstyle\bigcup_{i=1}^k T_{i,j}  \mid j \in
    \{1,2,\ldots,m\}\}$
  is a Tverberg $m$-partition of $P = \bigcup_{j=1}^{k} P_i$ with tolerance
  $t = \sum_{i=1}^k t_i +k -1$.
\end{lemma}
\begin{proof}
  Take $R \subseteq P$ with $|R|=t$. As
  $t = \sum_{i=1}^k t_i +k -1 < \sum_{i=1}^k (t_i + 1)$, there is an $i$ with
  $|{P}_i \cap R| \leq t_i$.
  Since $\mc{T}_i$ is $t_i$-tolerant, we have
  $\bigcap_{j=1}^{m}\conv(T_{i,j} \setminus R) \neq \emptyset$. Because each
  $T_{i,j}$ is contained in the corresponding set $T_j$ of $\mc{T}$,
  the convex hulls of the elements in $\mc{T}$ still intersect after the removal
  of $R$.
\end{proof}

From a mathematical perspective, the main motivation for
introducing tolerance to Tverberg partitions is the possibility
to achieve better bounds than by just combining regular Tverberg
partitions. This provides deeper insight in the intersection pattern
of convex sets. Nevertheless, Lemma~\ref{apr:lem:cb} is interesting
from an algorithmic viewpoint as it enables us to benefit from
existing approximation algorithms for regular Tverberg partitions
by implying a simple algorithm: compute regular Tverberg
partitions for disjoint subsets of $P$ and then merge them using
Lemma~\ref{apr:lem:cb}. This proves Proposition~\ref{apr:prop:cb}:

\begin{proposition}[Proposition~\ref{apr:prop:cb} restated]
  Let $P\subset \real^d$ and let $\mc{A}$ be an
  algorithm that computes a regular Tverberg $m$-partition
  for any point set of size $n_\mc{A}(m)$ in time $T_{\mc{A}}(m)$.
  Then, a $\left(\lfloor |P|/n_{\mc{A}}(m) \rfloor -1\right)$-tolerant Tverberg
  $m$-partition for $P$ can be computed in time
  $\mc{O}\left(T_{\mc{A}}(m) \cdot |P| / n_{\mc{A}}(m) \right)$.
\end{proposition}
\begin{proof}
  We split $P$ into $\lfloor |P|/n_{\mc{A}} \rfloor$ disjoint sets and use
  $\mc{A}$ to obtain for each subset a regular Tverberg partition. Applying
  Lemma~\ref{apr:lem:cb}, we obtain a $(\lfloor |P|/n_{\mc{A}} \rfloor
  -1)$-tolerant Tverberg $m$-partition. Since the merging step in
  Lemma~\ref{apr:lem:cb} takes linear time in $|P|$, the total running time is
  $\mc{O}\left(T_{\mc{A}}(m) \cdot |P| / n_{\mc{A}}(m) \right)$, as
  claimed.
\end{proof}

Table~\ref{apr:tab:conc} shows specific values for Proposition~\ref{apr:prop:cb}
applied to Miller \& Sheehy's and Mulzer \& Werner's algorithm.

\begin{table}
  \begin{center}
    \newcommand{\tableheader}[1]{\multicolumn{1}{|c|}{\textbf{#1}}}
    \renewcommand{\arraystretch}{1.5}
    \begin{tabular}{|l|c|c|}
      \hline
      \tableheader{Algorithm} &
      \tableheader{Tolerance} &
      \tableheader{Running time}
      \tabularnewline \hline \hline
      Proposition~\ref{apr:prop:cb} with Miller-Sheehy &
      $\lfloor |P|/2m(d+1)^2 \rfloor - 1$ &
      $m^{\bo{\log d}} d^{\bo{\log d}} |P|$
      \tabularnewline \hline
      Proposition~\ref{apr:prop:cb} with Mulzer-Werner &
      $\lfloor|P|/4m(d+1)^3\rfloor - 1$  &
      $d^{\bo{\log d}}|P|$
      \tabularnewline \hline
    \end{tabular}
  \end{center}
  \caption{Proposition~\ref{apr:prop:cb} applied to existing approximation
    algorithms for the regular Tverberg problem.}
  \label{apr:tab:conc}
\end{table}

\begin{remark}
Lemma~\ref{apr:lem:cb} gives a quick proof of a slightly weaker version of the
\Soberon{}-Strausz bound: partition $P$ into $t+1$ disjoint sets of size at least
$\lfloor |P| / (t+1)
\rfloor$. By Tverberg's theorem, for each subset there exists a Tverberg
partition with no tolerance of size
$\lceil \lfloor |P| / (t+1) \rfloor / (d+1) \rceil$. Using Lemma~\ref{apr:lem:cb},
we obtain a $t$-tolerant Tverberg partition of size
$\lceil \lfloor |P| / (t+1) \rfloor / (d+1) \rceil \geq \lceil |P| / (t+1)(d+1)
\rceil -1$ of $P$, which is at most one less than the \Soberon{}-Strausz bound.
This weaker bound was also stated by \GarciaColin{}~\cite{GarciaColin2007}.
Again, as already mentioned after Lemma~\ref{apr:lem:cb}, this is
interesting mostly from an algorithmic perspective since it implies
that computing slightly worse tolerant Tverberg partitions than
guaranteed by the \Soberon{}-Strausz bound is polynomial-time
equivalent to computing regular Tverberg partitions.
\end{remark}

\section{Hardness of Tolerance Testing}\label{sec:complexity}
Teng~\cite{Teng1992} proved that deciding whether a point is a
centerpoint (\textsc{TestingCenter}) is coNP-complete. We show the same for
deciding whether a Tverberg partition is $t$-tolerant
(\textsc{TestingTolerantTverberg}) by a reduction from \textsc{TestingCenter}.
The problems are formally defined as follows:

\begin{problem}[\textsc{TestingCenter}]
  \begin{enumerate}
    \item[]
      \begin{description}
    \item[Given] a point set $P\subset \real^d$, and a centerpoint candidate
      $c\in\real^d$, where $d$ is part of the input.
    \item[Decide] whether $c$ is a centerpoint of $P$.
  \end{description}
  \end{enumerate}
\end{problem}

\begin{problem}[\textsc{TestingTolerantTverberg}]
  \begin{enumerate}
    \item[]
      \begin{description}
    \item[Given] a point set $P\subset \real^d$, a
      partition $\mc{T}$ of $P$, and a conjectured tolerance $t\in\nat$, where
      $d$ is part of the input.
    \item[Decide] whether $\mc{T}$ is a $t$-tolerant Tverberg partition of $P$.
  \end{description}
  \end{enumerate}
\end{problem}

Note that the size of the partition $\mc{T}$ in the definition of
\textsc{TestingTolerantTverberg} can be constant.

The following lemma is folklore. We include the proof for completeness.
It is used in the reduction to connect the tolerance of a
Tverberg partition with the depth of points in the intersection of the convex
hulls.

\begin{lemma}\label{com:lem:depth}
  Let $P\subset\real^d$ and let $c\in\real^d$. Then $c$ has depth $t+1$ w.r.t.
  $P$ if and only if for all subsets $R\subset P$ with $|R| \leq t$, we have $c
  \in \conv(P \setminus R)$.
\end{lemma}
\begin{proof}
  We prove both directions by showing the contrapositive.\\
  \indent``$\Rightarrow$''
        Suppose there is some $R\subset P, |R|\leq t$ with $c\notin \conv(P
        \setminus R)$. Then, there is a
        half-space $h^+$ that contains $c$ but no points from $\conv(P \setminus R)$.
        Thus,
        $c\in h^+$ and $| P \cap h^+ | \leq |R| \leq t$,
        and hence $c$ has depth at most $t$ w.r.t. $P$.

    ``$\Leftarrow$''
        Assume $c$ has depth $t'\leq t$ w.r.t. $P$. Let $h^+$ be a half-space that
        contains $c$ and $t'$ points from $P$. Set $R=h^+ \cap P$. Then,
        $|R|\leq t$ and $c \notin \conv(P \setminus R)$.
\end{proof}

We are now ready to prove Theorem~\ref{thm:complexity}:
\begin{theorem}[Theorem~\ref{thm:complexity} restated]
  \textsc{TestingTolerantTverberg} is coNP-complete.
\end{theorem}
\begin{proof}
  We first check that \textsc{TestingTolerantTverberg} is indeed contained in
  coNP. Let $\mathcal{T}$ be a Tverberg partition of $P\subset
  \real^d$ that is claimed to have tolerance $t$. A witness to
  $\mathcal{T}$ not being a $t$-tolerant Tverberg partition is a subset
  $R\subseteq P$ of size at most $t$ such that
  $\bigcap_{T_i\in\mathcal{T}}\conv(T_i \setminus R) = \emptyset$. Checking if
  $R$ is a witness reduces to testing the feasibility  of the linear
  program defined by the following constraints for each element
  $T_i$ in \mc{T}:
  \begin{align*}
    \alpha_{i,1}\, p_{i,1} + \alpha_{i,2}\, p_{i,2} + \dots + \alpha_{i,|T_i
      \setminus R|}\, p_{i, |T_i \setminus R|} - x & = 0 \\
    \alpha_{i,1}+ \alpha_{i,2}+ \dots + \alpha_{i,|T_i\setminus R|} & = 1 \\
    \forall j \in \{1,2,\ldots,|T_i \setminus R|\}: \alpha_{i,j} & \geq 0 \qquad, &
  \end{align*}
  where $p_{i,j}$ denotes the $j$th point in $T_i \setminus R$.
  The linear program is feasible if and only if $\bigcap_{T_i\in\mc{T}}
  \conv(T_i \setminus R) \neq \emptyset$, i.e., if $R$ is not a witness.
  Since the number of constraints and variables is polynomial in the input size,
  feasibility checking of the above linear program can be carried out in
  polynomial time.

  Let $(P\subset\real^{d}, c\in\real^{d})$ be an input to
  \textsc{TestingCenter}. We embed the vector space $\real^{d}$ in $\real^{d+1}$
  by identifying it with the hyperplane $h: x_{d+1} = 0$.  Define $t=\lceil |P|
  / (d+1) \rceil -1$ and let $\ell$ be the line that is orthogonal to $h$ and
  passes through $c$. Furthermore, let $T^-$ and $T^+$ be sets of $t+1$
  arbitrary points in $\ell\cap h^-$ and $\ell\cap h^+$, respectively.  Set $T =
  T^- \cup T^+$. We claim that $\{P, T\}$ is a Tverberg 2-partition for $P\cup
  T$ with tolerance $t$ if and only if $c$ is a centerpoint of $P$; see
  Figure~\ref{com:fig:red}.

  ``$\Rightarrow$'' Assume $\{P, T\}$ is a $t$-tolerant Tverberg
    2-partition. By construction of $T$, we have $\conv(P) \cap \conv(T)  = \{c\}$.
    Thus, $c$ lies in the intersection
    of both convex hulls even if an arbitrary subset of size at most $t$ is removed.
    Lemma~\ref{com:lem:depth} implies that $c$ has depth $t+1  = \lceil |P| / (d+1)
    \rceil$ w.r.t. $P$, so $c$ is a centerpoint for $P$.

  ``$\Leftarrow$'' Assume $c$ is a centerpoint for $P$. By
      definition, $c$ has depth at least
      $\lceil|P|/(d+1)\rceil = t+1$ w.r.t. $P$. Lemma~\ref{com:lem:depth}
      then implies that $c$ is contained in the
      convex hull of $P$ even if any $t$ points from $P$ are removed. Since
      $T$ contains $t+1$ points on both sides of a line through $c$, $c$ is also
      contained in $\conv(T)$ if any $t$ points from $T$ are removed.
      Thus,
      $\{P, T\}$ is a $t$-tolerant Tverberg 2-partition for $P\cup T$.
\end{proof}
\begin{figure}[htbp]
  \begin{center}
    \includegraphics[width=0.5\textwidth]{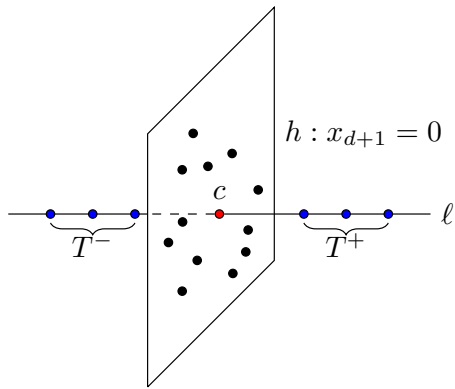}
  \end{center}
  \caption{Reduction of \textsc{TestingCenter} to
    \textsc{TestingTolerantTverberg}}
  \label{com:fig:red}
\end{figure}

\section{Conclusion}
We have shown that for each set $P\subset\real$ of size $m(t+2)-1$, a
$t$-tolerant Tverberg partition of size $m$ can be found in time
$\bo{mt\log t}$. The bound on the size of $P$ is tight, and it improves
the \Soberon{}-Strausz bound in one dimension.
Combining this with a lifting method, we could also get improved bounds in
two dimensions and an efficient algorithm for tolerant Tverberg partitions in
any fixed dimension.
However, the running time is exponential in the dimension.

This motivated us to look for a way of reusing the existing technology for
the regular Tverberg problem.
We have presented a reduction to the regular Tverberg problem
that enables us to
reuse the approximation algorithms by Miller \& Sheehy and Mulzer \& Werner.

Finally, we proved that testing whether a given Tverberg partition is of
some tolerance $t$ is coNP-complete.
Unfortunately, this does not imply anything about the complexity of
finding tolerant Tverberg partitions. It is not even clear
whether computing tolerant Tverberg partitions is harder than computing
regular Tverberg partitions. However, we could show that computing tolerant
Tverberg partitions with smaller tolerance than guaranteed by the
\Soberon{}-Strausz bound is polynomial-time equivalent to computing regular
Tverberg partitions.

It remains open whether the bound by \Soberon{} and Strausz is tight for $d > 2$.
We believe that our results in one and two dimensions indicate that
the bound can be improved also in general dimension.
Another open problem is finding a \emph{pruning strategy} for tolerant
Tverberg partitions. By this, we mean an
algorithm that efficiently reduces the sizes of the sets in a $t$-tolerant Tverberg
partition without deteriorating the tolerance.
Such a pruning strategy could be used to improve the quality of our algorithms.
In Miller \& Sheehy's and Mulzer \& Werner's algorithms,
\Caratheodory{}'s theorem was used for this task. Unfortunately, this result does not
preserve the tolerance of the pruned partitions. The generalized tolerant
\Caratheodory{} theorem~\cite{Montejano2011} also does not seem to help. It remains an
interesting problem to develop criteria for superfluous points
in tolerant Tverberg partitions.
\\

\noindent\textbf{Acknowledgments.}
We would like to thank the anonymous reviewers for their helpful
and detailed comments that helped to improve the quality of the
paper. In particular, we would like to thank an anonymous referee
for pointing out that the algorithm in
Proposition~\ref{apr:prop:cb} could be greatly simplified.

\bibliographystyle{abbrv}
\bibliography{library}
\end{document}